\documentclass[11pt]{article}
\usepackage{times}
\usepackage{url}
\usepackage{amsmath,amsfonts,amsthm,amssymb,multirow,dsfont,etex}
\usepackage{times}
\usepackage{paralist}
\usepackage{graphicx}
\usepackage{floatpag}
\usepackage{algorithm}
\usepackage[noend]{algpseudocode}
\usepackage{bbold}
\usepackage{enumitem}
\usepackage{subcaption} 
\usepackage{soul}
\usepackage{comment}
\usepackage{calc}
\usepackage{upgreek}
\usepackage{mathtools}
\usepackage{color}
\usepackage{xspace}
\usepackage{fullpage}
\allowdisplaybreaks

\DeclarePairedDelimiter\abs{\lvert}{\rvert}%
\def \F {{\mathbb F}}

\def \poly {\text{poly}}

\def \NTIME {\text{NTIME}}
\def \EMAJ {\text{EMAJ}}
\def \ETHR {\text{ETHR}}
\def \THR {\text{THR}}

\def \SUM {\text{SUM}}
\def \AND {\text{AND}}
\def \OR {\text{OR}}
\def \PSUM {\text{SUM}^{\geq 0}}
\def \ENC {\text{ENC}}
\def \eps {{\varepsilon}}

\def \size {\text{size}}

\newcommand{\NC}{\mathsf{NC}\xspace}
\def \P {\mathsf{P}\xspace}
\newcommand{\SIZE}{\mathsf{SIZE}\xspace}
\newcommand{\NP}{\mathsf{NP}\xspace}
\newcommand{\TC}{{\sf TC}}
\newcommand{\qNP}{\mathsf{Quasi}\text{-}\NP\xspace}
\newcommand{\AC}{\mathsf{AC}}
\newcommand{\SYM}{\mathsf{SYM}}
\newcommand{\ACC}{\mathsf{ACC}}
\newcommand{\NEXP}{\mathsf{NEXP}\xspace}

\newcommand{\NQP}{\mathsf{quasi}\text{-}\NP\xspace}
\newcommand{\PTIME}{\mathsf{P}\xspace}
\newcommand{\Ppoly}{\PTIME_{\sf/poly}}

\newtheorem{theorem}{Theorem}[section]

\newtheorem{corollary}{Corollary}[section]
\newenvironment{proofof}[1]{\noindent {\bf Proof of #1.  }}{\hfill$\Box$}
\newtheorem{definition}{Definition}[section]
\newtheorem{lemma}{Lemma}[section]
\newtheorem{claim}{Claim}

\theoremstyle{definition}
\theoremstyle{plain}
\newtheorem{conjecture}{Conjecture}

\title{Lower Bounds Against Sparse Symmetric Functions of ACC Circuits: Expanding the Reach of $\#$SAT Algorithms} 

\author{Nikhil Vyas\footnote{\texttt{nikhilv@mit.edu}, Supported by NSF CCF-1909429.}\\MIT \and Ryan Williams\footnote{\texttt{rrw@mit.edu}, Supported by NSF CCF-1741615 and NSF CCF-1909429.}\\MIT}

\begin{document}

\maketitle

\begin{abstract}
We continue the program of proving circuit lower bounds via circuit satisfiability algorithms. So far, this program has yielded several concrete results, proving that functions in $\qNP = \NTIME[n^{(\log n)^{O(1)}}]$ and $\NEXP$ do not have small circuits (in the worst case and/or on average) from various circuit classes ${\cal C}$, by showing that ${\cal C}$ admits non-trivial satisfiability and/or $\#$SAT algorithms which beat exhaustive search by a minor amount.

In this paper, we present a new strong lower bound consequence of non-trivial $\#$SAT algorithm for a circuit class ${\mathcal C}$. Say a symmetric Boolean function $f(x_1,\ldots,x_n)$ is \emph{sparse} if it outputs $1$ on $O(1)$ values of $\sum_i x_i$. We show that for every sparse $f$, and for all ``typical'' ${\cal C}$, faster $\#$SAT algorithms for ${\cal C}$ circuits actually imply lower bounds against the circuit class $f \circ {\cal C}$, which may be \emph{stronger} than ${\cal C}$ itself. In particular:
\begin{itemize}
    \item $\#$SAT algorithms for $n^k$-size ${\cal C}$-circuits running in $2^n/n^k$ time (for all $k$) imply $\NEXP$ does not have $f \circ {\cal C}$-circuits of polynomial size.
    \item $\#$SAT algorithms for $2^{n^{\eps}}$-size ${\cal C}$-circuits running in $2^{n-n^{\eps}}$ time (for some $\eps > 0$) imply $\qNP$ does not have $f \circ {\cal C}$-circuits of polynomial size.
\end{itemize}
Applying $\#$SAT algorithms from the literature, one immediate corollary of our results is that $\qNP$ does not have $\EMAJ \circ \ACC^0 \circ \THR$ circuits of polynomial size, where $\EMAJ$ is the ``exact majority'' function, improving previous lower bounds against $\ACC^0$ [Williams JACM'14] and $\ACC^0 \circ \THR$ [Williams STOC'14], [Murray-Williams STOC'18]. This is the first nontrivial lower bound against such a circuit class.
\end{abstract}

\newpage

\section{Introduction}

Currently, our knowledge of algorithms vastly exceeds our knowledge of lower bounds. Is it possible to bridge this gap, and use the existence of powerful algorithms to give lower bounds for hard functions? Over the last decade, the program of proving lower bounds via algorithms has been positively addressing this question. A line of work starting with Kabanets and Impagliazzo~\cite{KabanetsI04} has shown how deterministic subexponential-time algorithms for polynomial identity testing would imply lower bounds against arithmetic circuits. 
Starting around 2010~\cite{Williams-SICOMP13,Williams-jacm14}, it was shown that even \emph{slightly nontrivial} algorithms could imply Boolean circuit lower bounds. For example, a circuit satisfiability algorithm running in $O(2^n/n^k)$ time (for all $k$) on $n^k$-size circuits with $n$ inputs would already suffice to yield the (infamously open) lower bound $\NEXP \not\subset \P/\poly$. More generally, a generic connection was found between non-trivial SAT algorithms and circuit lower bounds:

\begin{theorem}[\cite{Williams-SICOMP13,Williams-jacm14}, Informal]\label{w-thm} Let ${\cal C}$ be a circuit class closed under AND, projections, and compositions.\footnote{It is not necessary to know precisely what these conditions mean, as we will use different conditions in our paper anyway. The important point is that these conditions hold for most interesting circuit classes that have been studied, such as $\AC^0$, $\TC^0$, $\NC^1$, $\NC$, and general fan-in two circuits.} Suppose for all $k$ there is an algorithm $A$ such that, for every ${\cal C}$-circuit of $n^k$ size, $A$ determines its satisfiability in $O(2^n/n^k)$ time. Then $\NEXP$ does not have polynomial-size ${\cal C}$-circuits.
\end{theorem}

To illustrate Theorem~\ref{w-thm} with two examples, when ${\cal C}$ is the class of general fan-in 2 circuits, Theorem~\ref{w-thm} says that non-trivial Circuit SAT algorithms imply $\NEXP \not\subset \P/\poly$; when ${\cal C}$ is the class of Boolean formulas, it says non-trivial Formula-SAT algorithms imply $\NEXP \not\subset \NC^1$. Both are major open questions in circuit complexity. Theorem~\ref{w-thm} and related results have been applied to prove several concrete circuit lower bounds: super-polynomial lower bounds for $\ACC^0$~\cite{Williams-jacm14}, $\ACC^0 \circ \THR$~\cite{Williams-TOC18}, quadratic lower bounds for depth-two symmetric and threshold circuits~\cite{Tamaki16,AlmanCW16}, and average-case lower bounds as well~\cite{ChenOS18,Chen19}.

Recently, the algorithms-to-lower-bounds connection has been extended to show a trade-off between the running time of the SAT algorithm on large circuits, and the complexity of the hard function in the lower bound. In particular, it is even possible in principle to obtain circuit lower bounds against $\NP$ with this algorithmic approach.

\begin{theorem}[\cite{MurrayW18}, Informal] \label{mw-thm} Let ${\cal C}$ be a class of circuits closed under unbounded AND, ORs of fan-in two, and negation. Suppose there is an algorithm $A$ and $\eps > 0$ such that, for every ${\cal C}$-circuit $C$ of $2^{n^{\eps}}$ size, $A$ solves satisfiability for $C$ in $O(2^{n-n^{\eps}})$ time. Then $\qNP$ does not have polynomial-size ${\cal C}$-circuits.\footnote{In this paper, we use the notation $\qNP := \bigcup_k \NTIME[n^{(\log n)^k}]$.}
\end{theorem}

In fact, Theorem~\ref{mw-thm} holds even if $A$ only distinguishes between unsatisfiable circuits from those with at least $2^{n-1}$ SAT assignments; we call this easier problem GAP-UNSAT. 

Intuitively, the aforementioned results show that as the circuit satisfiability algorithms improve in running time and scope, they imply stronger lower bounds. In all known results, to prove a lower bound against ${\cal C}$, one must design a SAT algorithm for a circuit class that is at least as powerful as ${\cal C}$. Inspecting the proofs of the above theorems carefully, it is not hard to show that, even if ${\cal C}$ did not satisfy the desired closure properties, it would suffice to give a SAT algorithm for a slightly more powerful class than the lower bound. For example, in Theorem~\ref{mw-thm}, a SAT algorithm running in $O(2^{n-n^{\eps}})$ time for $2^{n^{\eps}}$-size AND of ORs of three (possibly negated) ${\cal C}$ circuits (on $n$ inputs, of $2^{n^{\eps}}$ size) would still imply ${\cal C}$-circuit lower bounds for $\qNP$. Our key point here is that \emph{these proof methods require a SAT algorithm for a potentially more powerful circuit class than the class for which we can conclude a lower bound.} A compelling question is whether this requirement is an artifact of our proof method, or is it inherent?

\paragraph*{Lower bounds for more powerful classes from SAT algorithms?} We feel it is natural to conjecture that a SAT algorithm for a circuit class ${\cal C}$ implies a lower bound against a class that is \emph{more powerful} than ${\cal C}$, because checking satisfiability is itself a very powerful ability. Intuitively, a non-trivial SAT algorithm for ${\cal C}$ on $n$-input circuits is computing a \emph{uniform OR} of $2^n$ ${\cal C}$-circuits evaluated on fixed inputs, in $o(2^n)$ time. (Recall that a ``uniform'' circuit informally means that any gate of the circuit can be efficiently computed by an algorithm.) If there were an algorithm to decide the outputs of uniform ORs of ${\cal C}$-circuits more efficiently than their actual circuit size, perhaps this implies a lower bound against $\OR \circ {\cal C}$ circuits.

Similarly, a $\#$SAT algorithm for ${\cal C}$ on $n$-input circuits can be used to compute the output of any circuit of the form $f(C(x_1),\ldots,C(x_{2^n}))$ where $f$ is a uniform symmetric Boolean function, $C$ is a ${\cal C}$-circuit with $n$ inputs, and $x_1,\ldots,x_{2^n}$ is an enumeration of all $n$-bit strings. Should we therefore expect to prove lower bounds on symmetric functions of ${\cal C}$-circuits, using a $\#$SAT algorithm? This question is particularly significant because in many of the concrete lower bounds proved via the program~\cite{Williams-jacm14,Williams-TOC18,MurrayW18}, non-trivial $\#$SAT algorithms were actually obtained, not just SAT algorithms. So our question amounts to asking: \emph{how strong of a circuit lower bound we can prove, given the SAT algorithms we already have?} We use $\SYM$ to denote the class of Boolean symmetric functions.

\begin{conjecture}[$\#$SAT Algorithms Imply Symmetric Function Lower Bounds, Informal] \label{conj} Non-trivial $\#$SAT algorithms for circuit classes ${\cal C}$ imply size lower bounds against $\SYM \circ {\cal C}$ circuits. In particular, all statements in Theorem~\ref{w-thm} and Theorem~\ref{mw-thm} hold when the SAT algorithm is replaced by a $\#$SAT algorithm, and the lower bound consquence for ${\cal C}$ is replaced by $\SYM \circ {\cal C}$.
\end{conjecture}

If Conjecture~\ref{conj} is true, then existing $\#$SAT algorithms would already imply super-polynomial lower bounds for $\SYM \circ \ACC^0 \circ \THR$ circuits, a class that contains depth-two symmetric circuits (for which no lower bounds greater than $n^2$ are presently known)~\cite{Tamaki16,AlmanCW16}.

More intuition for Conjecture~\ref{conj} can be seen from a recent paper of the second author, who showed how $\#$SAT algorithms for a circuit class ${\cal C}$ can imply lower bounds on \emph{(real-valued) linear combinations of ${\cal C}$-circuits}~\cite{Williams-CCC18}. For example, known $\#$SAT algorithms for $\ACC^0$ circuits imply $\qNP$ problems cannot be computed via polynomial-size linear combinations of polynomial-size $\ACC^0 \circ \THR$ circuits. However, the linear combination representation is rather constrained: the linear combination is required to always output $0$ or $1$. Applying PCPs of proximity, Chen and Williams~\cite{ChenW19} showed that the lower bound of~\cite{Williams-CCC18} can be extended to ``approximate'' linear combinations of ${\cal C}$-circuits, where the linear combination does not have to be exactly $0$ or $1$, but must be closer to the correct value than to the incorrect one, within an additive constant factor. These results show, in principle, how a $\#$SAT algorithm for a circuit class ${\cal C}$ can imply lower bounds for a stronger class of representations than ${\cal C}$.

\subsection{Conjecture~\ref{conj} Holds for Sparse Symmetric Functions}

In this paper, we take a concrete step towards realizing Conjecture~\ref{conj}, by proving it for ``sparse'' symmetric functions. We say a symmetric Boolean function $f(x_1,\ldots,x_n)$ is \emph{$k$-sparse} if $f$ is $1$ on at most $k$ values of $\sum_i x_i$. The $1$-sparse symmetric functions are called the \emph{exact threshold} ($\ETHR$ with polynomial weights) or \emph{exact majority} ($\EMAJ$) functions, which have been studied for years in both circuit complexity~(e.g.~\cite{Green00,BeigelTT92,Hansen07,Hansen09,HansenP10}) and structural complexity theory, where the corresponding complexity class (computing an exact majority over all computation paths) is known as $\mathsf{C}_{=}\mathsf{P}$~\cite{Wagner86}. 

\begin{theorem}\label{thm:main} Let ${\cal C}$ be closed under $\AND_2$, negation, and suppose the all-ones and parity function are in ${\cal C}$. Let $f=\{f_n\}$ be a family of $k$-sparse symmetric functions for some $k = O(1)$. 
\begin{itemize}
    \item If there is a $\#$SAT algorithm for $n^k$-size ${\cal C}$-circuits running in $2^n/n^k$ time (for all $k$), then $\NEXP$ does not have $f \circ {\cal C}$-circuits of polynomial size.
    \item If there is a $\#$SAT algorithm for $2^{n^{\eps}}$-size ${\cal C}$-circuits running in $2^{n-n^{\eps}}$ time (for some $\eps > 0$), then $\qNP$ does not have $f \circ {\cal C}$-circuits of polynomial size. 
\end{itemize}
\end{theorem}

Applying known $\#$SAT algorithms for $\AC^0[m] \circ \THR$ circuits from~\cite{acc-algo}, we obtain:

\begin{corollary} For all constant depths $d \geq 2$ and constant moduli $m \geq 2$, $\qNP$ does not have polynomial-size $\EMAJ \circ \AC^0[m] \circ \THR$ circuits.
\end{corollary}

\subsection{Intuition}\label{subsec:int} Here we briefly explain the new ideas that lead to our new circuit lower bounds. 

As in prior work~\cite{Williams-CCC18,ChenW19}, the high-level idea is to show that if (for example) $\qNP$ has polynomial-size $\EMAJ \circ \mathcal{C}$, and there is a $\#$SAT algorithm for ${\cal C}$ circuits, then we can design a nondeterministic algorithm for verifying GAP Circuit Unsatisfiability (GAP-UNSAT) on generic circuits that beats exhaustive search. In GAP-UNSAT, we are given a generic circuit and are promised that it is either unsatisfiable, or at least half of its possible assignments are satisfying, and we need to nondeterministically prove the unsatisfiable case. (Note this is a much weaker problem than SAT.) As shown in~\cite{Williams-SICOMP13,Williams-jacm14,MurrayW18}, combining a nondeterministic algorithm for GAP-UNSAT with the hypothesis that $\qNP$ has polynomial-size circuits, we can derive that nondeterministic time $2^n$ can be simulated in time $o(2^n)$, contradicting the nondeterministic time hierarchy theorem. 

Our key idea is to use probabilistically checkable proofs (PCPs) in a new way to exploit the power of a $\#$SAT algorithm. First, let's observe a task that a $\#$SAT algorithm for ${\cal C}$ can compute on an $\EMAJ \circ {\cal C}$ circuit. Suppose our $\EMAJ \circ {\cal C}$ circuit has the form
\[D(x) = \left[\sum_{i=1}^t C_i(x) = s\right],\] where each $C_i(x)$ is a Boolean ${\cal C}$-circuit on $n$ inputs, $s$ is a threshold value, and our circuit outputs $1$ if and only if the sum of the $C_i$'s equals $s$.\footnote{We are using the standard Iverson bracket notation, where $[P]$ is $1$ if predicate $P$ is true, and $0$ otherwise.} Consider the expression 
\begin{align}\label{def-of-E}
   E(x) := \left(\sum_{i=1}^t C_i(x) - s\right)^2.
\end{align} Treated as a function, $E(x)$ outputs integers; $E(a) = 0$ when $D(a) = 1$, and otherwise $E(a) \in [1,(t+s)^2]$. We first claim that the quantity 
\begin{align}\label{sum-for-E}
    \sum_{a \in \{0,1\}^n} E(a)
\end{align} can be compute faster than exhaustive search using a faster $\#$SAT algorithm. To see this, using distributivity, we can rewrite \eqref{def-of-E} as \[E(x) = \sum_{i,j} (C_i \wedge C_j)(x) - 2 s \sum_i C_i(x) + s^2.\] Assuming ${\cal C}$ is closed under conjunction, each $C_i \wedge C_j$ is also a ${\cal C}$-circuit, and we can compute \[ \sum_{a \in \{0,1\}^n} E(a) = \sum_{i,j} \left(\sum_{a \in \{0,1\}^n}(C_i \wedge C_j)(a)\right) - 2 s \sum_i \left(\sum_{a \in \{0,1\}^n} C_i(a)\right) + s^2 \cdot 2^n\] by making $O(t^2)$ calls to a $\#$SAT algorithm. Thus we can compute \eqref{sum-for-E} using a $\#$SAT algorithm.

How is computing \eqref{sum-for-E} useful? This is where PCPs come in. We cannot use \eqref{sum-for-E} to directly solve $\#$SAT for $D$ (otherwise as $\#$SAT algorithms imply SAT algorithms we could apply existing work~\cite{Williams-jacm14}, and be done). But we can use \eqref{sum-for-E} to obtain a \emph{multiplicative approximation} to the number of assignments that \emph{falsify} $D$. In particular, each satisfying assignment is counted zero times in \eqref{sum-for-E}, and each falsifying assignment is counted between $1$ and (less than) $(t+s)^2$ times. We want to exploit this, and obtain a faster GAP-UNSAT algorithm. Given a circuit which is a GAP-UNSAT instance, we start by using an efficient hitting set construction \cite{hitting-set} to increase the gap of GAP-UNSAT. We obtain a new circuit $C(x)$ which is either UNSAT or has at least $2^n-o(2^n)$ satisfying assignments (Section~\ref{subsec:gap}). Next (Lemma~\ref{lem:ecc+pcpp}) we apply a PCP of Proximity and an error correcting code to $C$, yielding a 3-SAT instance over $x$ and extra variables, with constant gap (similar to Chen-Williams~\cite{ChenW19}), and we amplify this gap using standard serial repetition. Finally, we apply the FGLSS~\cite{FGLSS} reduction (Lemma~\ref{lem:fglss-par}) to the 3-SAT instance, obtaining Independent Set instances with a large gap between the YES case and NO case. In particular, for all inputs $x$, when $C(x)=1$ there is a large independent set in the resulting graph, and when $C(x)=0$, there are only small independent sets in the resulting graph (see Lemma~\ref{lem:comb}). Returning to the assumption that $\qNP$ has small $\EMAJ \circ {\cal C}$ circuits, and applying an easy witness lemma~\cite{MurrayW18}, it follows that the solutions to the independent set instance can be encoded by $\EMAJ \circ {\cal C}$ circuits. Because of the large gap between the YES case and NO case, our multiplicative approximation to the number of UNSAT assignments can be used to distinguish the unsatisfiable case and the ``many satisfying assignments'' case of GAP-UNSAT, which finishes the argument.

One interesting bottleneck is that we cannot \emph{directly} apply serial repetition and the FGLSS reduction in our argument; we need the PCP machinery we use to behave similarly on all inputs $x$ to the original circuit $C$. This translates to studying the behavior of these reductions \emph{with respect to partial assignments}. While for these two reductions we are able to prove that they behave ``nicely'' with respect to partial assignments, it is entirely unclear that this is true for other PCP reductions such alphabet reduction, parallel repetition, and so on.

Our approach is very general; to handle $k$-sparse symmetric functions, we can simply modify the function $E$ accordingly.

\section{Preliminaries and Organization}

We assume general familiarity with basic concepts in circuit complexity and computational complexity~\cite{Arora-Barak}. In particular we assume familiarity with $\AC^0$, $\ACC^0$, $\Ppoly$, $\NEXP$, and so on.

\paragraph*{Circuit Notation.} Here we define notation for the relevant circuit classes. By $\size_{\mathcal{C}}(h(n))$ we denote circuits from circuit class $\mathcal{C}$ with size at most $h(n)$.

\begin{definition}
An $\EMAJ \circ \mathcal{C}$ circuit (a.k.a. ``exact majority of ${\cal C}$ circuit'') has the general form $\EMAJ(C_1(x), C_2(x), \ldots, C_t(x), u)$, where $u$ is a positive integer, $x$ are the input variables, $C_i \in \mathcal{C}$, and the gate $\EMAJ(y_1,\ldots,y_t,u)$ outputs $1$ if and only if exactly $u$ of the $y_i$'s output 1. \end{definition}

\begin{definition}
A $\PSUM \circ \mathcal{C}$ circuit (``positive sum of ${\cal C}$ circuits'') has the form \[\PSUM(C_1(x), C_2(x), \ldots, C_t(x)) = \sum_{i \in [t]} C_i(x)\] where $C_i$ is either a $\mathcal{C}$-circuit or $-1$ times a $\mathcal{C}$-circuit and we are promised that $\sum_{i \in [t]} C_i(x) \geq 0$ over all $x \in \{0,1\}^n$.

Given a set of circuits $\{C_i\}$, we say that $f : \{0,1\}^n \rightarrow \{0,1\}$ is \emph{represented} by the positive-sum circuit $\PSUM(C_1(x), C_2(x), \ldots, C_t(x))$ if for all $x$, $f(x) = 1$ when $\sum_{i \in [t]} C_i(x) > 0$, and $f(x) = 0$ when $\sum_{i \in [t]} C_i(x) = 0$.
\end{definition}

\begin{definition}\label{def:typical}
A circuit class $\mathcal{C}$ is \emph{typical} if there is a $k > 0$ such that the following hold:
\begin{itemize}
\item {\bf Closure under negation.} For every $\mathcal{C}$ circuit $C$, there is a circuit $C'$ computing the negation of $C$ where $\size(C') \leq \size(C)^k$.
\item {\bf Closure under AND.} For every $\mathcal{C}$ circuits $C_1$ and $C_2$, there is a circuit $C'$ computing the AND of $C_1$ and $C_2$ where $\size(C') \leq (\size(C_1)+\size(C_2))^k$.
\item {\bf Contains all-ones.} The function ${\bf 1}_n : \{0,1\}^n \rightarrow \{0,1\}$ has a $\mathcal{C}$ circuit of size $O(n^k)$. 
\end{itemize}
\end{definition}

The vast majority of circuit classes that are studied ($\AC^0$, $\ACC^0$, $\TC^0$, $\NC^1$, $\Ppoly$) are typical.\footnote{A notable exception (as far as we know) is the class of depth-$d$ \emph{exact} threshold circuits for a fixed $d \geq 2$, because we do not know if such classes are closed under negation. Similarly, we do not know if the class of depth-$d$ threshold circuits is typical. (In that case, the only non-trivial property to check is closure under AND; we can compute the AND of two threshold circuits with a quasi-polynomial blowup using Beigel-Reingold-Spielman~\cite{pp-closed}, but not with a polynomial blowup.)}
The next lemma shows that the negation of an exact-majority of ${\cal C}$ circuit can be represented as a ``positive-sum'' of ${\cal C}$ circuit, if ${\cal C}$ is typical.

\begin{lemma}\label{lem:emaj-to-psum} Let $\mathcal{C}$ be typical. If a function $f$ has a $\EMAJ \circ \mathcal{C}$ circuit $D$ of size $s$, then $\neg f$ can be represented by a $\PSUM \circ \mathcal{C}$ circuit $D'$ of size $\poly(s)$. Moreover, a description of the circuit $D'$ can be obtained from a description of $D$ in polynomial time.
\end{lemma}

\begin{proof}
Suppose $f$ is computable by the $\EMAJ \circ \mathcal{C}$ circuit $D = \EMAJ(D_1, D_2, \ldots, D_t, u)$, where $u \in \{0,1,\ldots,t\}$. Consider the expression \[E(x) := (\SUM(D_1, D_2, \ldots, D_t)-u)^2.\] Note that $E(x)=0$ when $D(x) = 1$, and $E(x) > 0$ when $D(x) = 0$. So in order to prove the lemma, it suffices to show that $E$ can be written as a $\PSUM \circ \mathcal{C}$ circuit. Expanding the expression $E$,
\begin{align*}
E(x) &= \SUM(D_1, D_2, \ldots, D_t)^2 - 2u \cdot \SUM(D_1, D_2, \ldots, D_t) + u^2\\
&= \sum_{i,j = 1}^t (D_i \wedge D_j) - \sum_{j=1}^{2u} \sum_{i=1}^t D_i  + u^2.
\end{align*} By Definition~\ref{def:typical} $\AND_2 \circ \mathcal{C} = \mathcal{C}$, each $D_i \wedge D_j$ is a circuit from $\mathcal{C}$ of size $\poly(s)$. Since the all-ones function is in $\mathcal{C}$, the function $x \mapsto u^2$ also has a $\SUM \circ \mathcal{C}$ circuit of size $O(t^2)$. Therefore there are circuits $D'_i \in \mathcal{C}$ and $t' \leq O(t^2)$ such that by defining $D' := \PSUM(D'_1,\ldots,D'_{t'})$ we have $D'(x) = E(x)$ for all $x$. 
\end{proof}

\paragraph*{Error-Correcting Codes.} We will need a (standard) construction of binary error correcting codes with constant rate and constant relative distance.

\begin{theorem}[\cite{lin-codes}] \label{thm:code} There are universal constants $c \geq 1$ and $\delta \in (0,1)$ such that for all sufficiently large $n$, there are linear functions $ENC^n : (\F_2)^{n} \rightarrow (\F_2)^{cn}$ such that for all $x \neq y$ with $|x|=|y|=n$, the Hamming distance between $ENC^n(x)$ and $ENC^n(y)$ is at least $\delta n$.
\end{theorem}
In what follows, we generally drop the superscript $n$ for notational brevity. Note that each bit of output $\ENC^n_i(x)$ (for $i=1,\ldots,cn$) is a parity function on some subset of the input bits.
\subsection{Weak CAPP Algorithms Are Sufficient For Lower Bounds}\label{subsec:gap}

Murray and Williams~\cite{MurrayW18} showed that CAPP/GAP-UNSAT algorithms, i.e., algorithms which distinguish between unsatisfiable circuits and circuits with $\geq 2^{n-1}$ satisfying assignments are enough to give lower bounds. For our results, it is necessary to strengthen the ``gap'', which can be done using known hitting set constructions.

\begin{lemma}[Corollary C.5 in~\cite{hitting-set}, Hitting Set Construction] \label{lem:hit}
There is a constant $\psi > 0$ and a $\poly(n, \log{g})$ time algorithm $S$ such that, given a (uniform random) string $r$ of $n+\psi\cdot \log{g}$ bits, $S$ outputs $t=O(\log g)$ strings $x_1, x_2, \ldots, x_{t} \in \{0,1\}^n$ such that for every $f : \{0,1\}^n \rightarrow \{0,1\}$ with $\sum_x f(x) \geq 2^{n-1}$, $\Pr_r[\OR_{i=1}^t f(x_i) = 1] \geq 1-1/g$. 
\end{lemma}

We will use the following ``algorithms to lower bounds'' connections as black box: 

\begin{theorem}[\cite{MurrayW18}]\label{thm:nqp-lb-CAPP}
Suppose for some constant $\eps \in (0, 1)$ there is an algorithm $A$ that for all $2^{n^{\eps}}$-size circuits $C$ on $n$ inputs, $A(C)$ runs in $2^{n-n^{\eps}}$ time, outputs YES on all unsatisfiable $C$, and outputs NO on all $C$ that have at least $2^{n-1}$ satisfying assignments. Then for all $k$, there is a $c \geq 1$ such that $\NTIME[2^{\log^{ck^4/\eps} n}] \not\subset \SIZE[2^{\log^k n}]$.
\end{theorem}

Applying Lemma~\ref{lem:hit} to Theorem~\ref{thm:nqp-lb-CAPP}, we observe that the circuit lower bound consequence can be obtained from a significantly weaker-looking hypothesis. This weaker hypothesis will be useful for our lower bound results. 

\begin{theorem}\label{lem:nqp-lb-CAPP2}
Suppose for some constant $\eps \in (0, 1)$ there is an algorithm $A$ that for all $2^{n^{\eps}}$-size circuits $C$ on $n$ inputs, $A(C)$ runs in $2^{n}/g(n)^{\omega(1)}$ time, outputs YES on all unsatisfiable $C$, and outputs NO on all $C$ that have at least $2^{n}(1-1/g(n))$ satisfying assignments, for $g(n) = 2^{n^{2\eps}}$. Then for all $k$, there is a $c \geq 1$ such that $\NTIME[2^{\log^{ck^4/\eps} n}] \not\subset \SIZE[2^{\log^k n}]$.
\end{theorem}
\begin{proof} Our starting point is Theorem~\ref{thm:nqp-lb-CAPP} (\cite{MurrayW18}): we are given an $m$-input, $2^{m^{\delta}}$-size circuit $D'$ that is either UNSAT or has at least $2^{m-1}$ satisfying assignments, and we wish to distinguish between the two cases with a $2^{m-m^{\delta}}$-time algorithm. We set $\delta = \eps/2$
 
We create a new circuit $D$ with $n$ inputs, where $n$ satisfies \[n = m + \psi \cdot \log{g(n)},\] and $\psi > 0$ is the constant from Lemma~\ref{lem:hit}. (Note that, since $g(n)$ is time constructible and $g(n) \leq 2^{o(n)}$, such an $n$ can be found in subexponential time.) Applying the algorithm from Lemma~\ref{lem:hit}, $D$ treats its $n$ bits of input as a string of randomness $r$, computes $t=O(\log g(n))$ strings $x_1, x_2, \ldots, x_t \in \{0,1\}^m$ with a $\poly(m, \log g)$-size circuit, then outputs the OR of $D'(x_i)$ over all $i=1,\ldots,t$. Note the total size of our circuit $D$ is $\poly(m,\log g) + O(\log g)\cdot \size(D') = \poly(n) + O(n^{2\eps}) \cdot 2^{m^{\delta}} < 2^{n^{2\delta}} = 2^{n^{\eps}}$ as $\eps = 2\delta$.

Clearly, if $D'$ is unsatisfiable, then $D$ is also unsatisfiable. By Lemma~\ref{lem:hit}, if $D'$ has $2^{m-1}$ satisfying assignments, then $D$ has at least $2^{n}(1-1/g(n))$ satisfying assignments. As $\size(D) \leq 2^{n^{\eps}}$, by our assumption we can distinguish the case where $D$ is unsatisfiable from the case where $D$ has at least $2^{n}(1-1/g(n))$ satisfying assignments, with an algorithm running in time $2^{n}/g(n)^{\omega(1)}$. This yields an algorithm for distinguishing the original circuit $D'$ on $m$ inputs and $2^{m^{\delta}}$ size, running in time \[2^{n}/g(n)^{\omega(1)}
= 2^{m}g(n)^{O(1)}/g(n)^{\omega(1)} = 2^{m}/g(n)^{\omega(1)} \leq 2^{m}2^{-n^{2\eps}} \leq2^{m}2^{-n^{\delta}} \leq 2^{m-m^{\delta}},\] since $g(n) = 2^{n^{2\eps}}$. By Theorem~\ref{thm:nqp-lb-CAPP}, this implies that for all $k$, there is a $c \geq 1$ such that $\NTIME[2^{\log^{ck^4/\delta} n}] \not\subset \SIZE[2^{\log^k n}]$. As, $\eps = 2\delta$ we get that $\NTIME[2^{\log^{2ck^4/\eps} n}] \not\subset \SIZE[2^{\log^k n}]$. But as the constant $4$ can be absorbed in the constant $c$ hence we get that for all $k$, there is a $c \geq 1$ such that $\NTIME[2^{\log^{ck^4/\eps} n}] \not\subset \SIZE[2^{\log^k n}]$.
\end{proof}

\subsection{Organization}
In Section~\ref{sec:pcp} we give a reduction from Circuit SAT to ``Generalized'' Independent Set. Section~\ref{sec:main} uses this reduction to prove lower bounds for $\EMAJ \circ \mathcal{C}$ assuming \#SAT algorithms for $\mathcal{C}$ with running time $2^{n-n^{\eps}}$. Section~\ref{subsec:acc} uses this result to give lower bound for $\EMAJ \circ \ACC^0 \circ \THR$. Section~\ref{subsec:sparse} generalizes these results to $f \circ \mathcal{C}$ lower bounds where $f$ is a sparse symmetric function.  In Section~\ref{sec:np-nexp} we give lower bounds for $\EMAJ \circ \mathcal{C}$ assuming \#SAT algorithms for $\mathcal{C}$ with running time $2^{n}/n^{\omega(1)}$.

\section{From Circuit SAT to Independent Set}\label{sec:pcp}

The goal of this section is to give the main PCP reduction we will use in our new algorithm-to-lower-bound theorem. First we need a definition of ``generalized'' independent set instances, where some vertices have already been ``assigned'' in or out of the independent set.

\begin{definition}
Let $G = (V, E)$ be a graph. Let $\pi: V \to \{0, 1, *\}$ be a partial Boolean assignment to $V$. We define $G(\pi)$ to be a graph with the label function $\pi$ on its vertices (where each vertex gets the label $0$, or $1$, or no label). We construe $G(\pi)$ as an {\bf generalized independent set instance}, in which any valid independent set (vertex assignment) must be consistent with $\pi$: any independent set must contain all vertices labeled $1$, and no vertices labeled $0$.
\end{definition}

\begin{lemma}\label{lem:comb} Let $k$ be a function of $n$. 
Given a circuit $D$ on $X$ with $\abs{X} = n$ bits and of size $m > n$, there is a $\poly(m, 2^{O(k)})$-time reduction from $D$ to a generalized independent set instance on graph $G_{D} = (V_D, E_D)$, with the following properties. 
\begin{itemize}
\item Each vertex $v \in V_D$ is associated with a set of pairs $S_v$ of the form $\{(i, b)\} \subseteq [O(n)] \times \{0, 1\}$. The set $\{S_v\}$ is produced as part of the reduction.
\item Each assignment $x$ to $X$ defines a partial assignment $\pi_{x}$ to $V_D$ such that 
\[\pi_{x}(v) = \begin{cases}
0 &  ~\text{if} ~~ \exists (i, b) \in S_v \text{~such that~} ENC_i(x) \neq b\\
* &  ~\text{otherwise},
 \end{cases}\] 
where $ENC$ is the error-correcting code from Theorem~\ref{thm:code}.
\item If $D(x) = 0$, the maximum independent set in $G_{D}(\pi_{x})$ equals $\kappa$ for an integer $\kappa$, and furthermore given $x$, it can be found in time $\poly(n, m, 2^{O(k)})$. 
\item If $D(x) = 1$, then the maximum independent set in $G_{D}(\pi_{x})$ has size at most $\kappa/2^k$. 
\end{itemize}

\end{lemma}

Intuitively, the use of Lemma~\ref{lem:comb} is that we will start with a ``no satisfying assignment'' vs ``most assignments are satisfying'' GAP-UNSAT instance from Theorem~\ref{lem:nqp-lb-CAPP2}. Now in the ``no satisfying assignment'' case for all $x$ the reduced independent set instance $G_{D}(\pi_{x})$ has a large independent set instance. Counting the sum of independent sets over $x$ gives a high value. On the other hand in the `most assignments are satisfying'' case for most $x$ the reduced independent set instance $G_{D}(\pi_{x})$ has a small independent set and for a very few $x$, $G_{D}(\pi_{x})$ can have a large independent set. Hence in this case counting the sum of independent sets over all $x$ gives a low value. The difference between the high value and low value is big enough that even a approximate counting of these values as outlined in Section~\ref{subsec:int} is enough to distinguish and hence solve the GAP-UNSAT instance.

The remainder of this section is devoted to the proof of Lemma~\ref{lem:comb}.

Let us set up some notation for variable assignments to a formula. Let $F$ be a SAT instance on a variable set $Z$, and let $\tau: Z \to \{0, 1, \star\}$ be a partial assignment to $Z$. Then we define $F(\tau)$ to be the formula obtained by setting the variables in $F$ according to $\tau$. Note that we do not perform further reduction rules on the clauses in $F(\tau)$: for each clause in $F$ that becomes false (or true) under $\tau$, there is a clause in $F(\tau)$ which is always false (true).

For every subsequence $Y$ of variables from $Z$, and every vector $y \in \{0,1\}^{|Y|}$, we define $F(Y=y)$ to be the formula $F$ in which the $i^{th}$ variable in $Y$ is assigned $y_i$, and all other variables are left unassigned. 


\begin{lemma}[PCPP+ECC,~\cite{ChenW19}]\label{lem:ecc+pcpp} There is a polynomial-time transformation that, given a circuit $D$ on $n$ inputs of size $m \geq n$, outputs a 3-SAT instance $F$ on the variable set $Y \cup Z$, where $\abs{Y} \leq \poly(n)$, $\abs{Z} \leq \poly(m)$, and the following hold for all $x \in \{0,1\}^n$:

\begin{itemize}
    \item If $D(x) = 0$ then $F(Y = \ENC(x))$ on variable set $Z$ has a satisfying assignment $z_x$. Furthermore, there is a $\poly(m)$-time algorithm that given $x$ outputs $z_x$.
    \item if $D(x) = 1$ then there is no assignment to the $Z$ variables in $F(Y = \ENC(x))$ satisfying more than a $(1-\Omega(1))$-fraction of the clauses.
\end{itemize}

where $\ENC: \{0, 1\}^n \to \{0, 1\}^{O(n)}$ is the linear encoding function from Theorem~\ref{thm:code}. As it is a linear function, the $i^{th}$ bit of output $\ENC_i(x)$ satisfies $\ENC_i(x) = \oplus_{j \in U_i} x_j$ for some set $U_i$. 
\end{lemma}

Serial Repetition~\cite{ser-rep} is a basic operation on CSPs/PCPs, in which a new CSP is created whose constraints are ANDs of $k$ uniformly sampled clauses from the original CSP. Serial repetition is usually done for the purpose of reducing soundness, i.e., reducing the fraction of satisfiable clauses. We now state a derandomized version of serial repetition.

\begin{lemma}[Serial repetition~\cite{ser-rep}]\label{lem:ser-rep}
Given a 3-SAT instance $F$ on $n$ variables denoted by $Y$ with $m$ clauses we can construct a $O(k)$-SAT formula $F'$ on the same $n$ variables with $m2^{O(k)}$ clauses such that:
\begin{enumerate}
    \item If $Y = y$ satisfies $F$ then $y$ satisfies $F'$.\label{prop:sr1}
    \item If $F(Y = y)$ is at most $1-\Omega(1)$ satisfiable then $F'(Y = y)$ is at most $1/2^{k}$ satisfiable. \label{prop:sr2}
\end{enumerate}
\end{lemma}

Next we prove a stronger version of derandomized serial repetition with guarantees for partial assignments. The proof directly follows from the guarantees of standard Serial Repetition (Lemma~\ref{lem:ser-rep}).

\begin{lemma}[Serial repetition with partial assignments]\label{lem:ser-rep-par}
Let $k$ be a function of $n$. Given a 3-SAT instance $F$ on $n$ variables denoted by $Y, Z$ with $m$ clauses we can construct a $O(k)$-SAT formula $F'$ on the same $n$ variables with $m \cdot 2^{O(k)}$ clauses such that:
\begin{enumerate}
    \item If $Y, Z = y, z$ satisfies $F$ then $y, z$ satisfies $F'$. \label{prop:srp1}
    \item If $F(Y = y)$ is at most $1-\Omega(1)$ satisfiable then $F'(Y = y)$ is at most $1/2^{k}$ satisfiable.
\end{enumerate}
\end{lemma}
\begin{proof}
We prove that just standard serial repetition from Lemma~\ref{lem:ser-rep} suffices for proving this stronger property.

Property~\ref{prop:srp1} directly follows from Property~\ref{prop:sr1} in Lemma~\ref{lem:ser-rep}. 

Define $F_{y} = F(Y = y)$ where we treat any clauses that became FALSE or TRUE under $Y = y$ as normal clauses. Let $F'_{y}$ be the $O(k)$-SAT formula obtained by applying serial repetition to $f_{y}$ from Lemma~\ref{lem:ser-rep-par}.

In Serial Repetition~\cite{ser-rep} it is clear that clauses in $F'$ are just ANDs of clauses in $F$ and which clauses are part of the ``AND'' is only dependent on their index.

Due to this $F'(Y = y)$ i.e. first applying serial repetition then setting $Y = y$ is equivalent to first setting $Y = y$ and then applying serial repetition i.e. $F'_{y}$. 

By our assumption $F_{y}$ is at most $1-\Omega(1)$ satisfiable and hence by Property~\ref{prop:sr2} of Lemma~\ref{lem:ser-rep} $F'_{y}$ is at most $1/2^k$ satisfiable. As $F'_{y} = F'(Y = y)$ we have that $F'(Y = y)$ is at most $1/2^k$ satisfiable.
\end{proof}

The FGLSS reduction~\cite{FGLSS} maps a CSP $\Phi$ to a graph $G_{\Phi}$ such that the MAX-SAT value in $\Phi$ is equal to the size of the maximum independent set in $G_{\Phi}$.
\begin{lemma}[FGLSS~\cite{FGLSS}]\label{lem:fglss}
Let $F$ be a $k$-SAT instance on variable set $Y$ with $\abs{Y} = n$ and $m$ clauses. There exists a $\poly(n, m, 2^{O(k)})$ time reduction graph from $F$ to a graph $G_{F} = (V_F, E_F)$ such that: the size of maximum independent set in $G_{F}$ is exactly equal to maximum clauses satisfiable in $F$.
\end{lemma}

 We note that a stronger version of the FGLSS reduction~\cite{FGLSS} holds with guarantees for partial assignments. The proof is very similar to the proof of the standard FGLSS reduction (Lemma~\ref{lem:fglss}).

\begin{lemma}[FGLSS with partial assignments]\label{lem:fglss-par}
Let $F$ be a $k$-SAT instance on variable set $Y, Z$ with $\abs{Y}+\abs{Z} = n$ and $m$ clauses. There exists a $\poly(n, m, 2^{O(k)})$ time reduction graph from $F$ to an independent set instance on graph $G_{F} = (V_F, E_F)$. Each vertex $v \in V_F$ is a associated to a set $T_v$ of $(i \in [\abs{Y}], b \in \{0, 1\})$ pairs. For each partial assignment of the form $\tau : Y \to \{0, 1\}$ define a partial assignment $\pi_{\tau}$ to $V_F$ such that: 
\[\pi_{\tau}(v) = \begin{cases}
0 &  ~\text{if} ~~ \exists (i, b) \in T_v \text{~such that~} \tau(Y_i) \neq b\\
* &  ~\text{otherwise},
 \end{cases}\] 
Then the max independent set in $G_{F}(\pi_{\tau})$ equals the max number of clauses satisfiable in $F(\tau)$.
\end{lemma}
\begin{proof}
Let $w$ be a clause in $F$ and $w_i$ denote the $i^{th}$ variable in $w$. Let $\ell$ denote a satisfying assignment to $w$. For every $w, \ell$ pair create a vertex in $V_{F}$. Let $v$ be the vertex associated with a particular $w, \ell$. Let $T_{v} = \{(w_i, \ell_i\}$ represent the assignment $w_i = \ell_i$ for $1 \leq i \leq k$.

Make an edge between vertex $u$ and vertex $v$ if the assignment $T_u$ and $T_v$ contradict each other. Note that this means that there is always an edge between two vertices associated to the same clause but different satisfying assignments i.e. vertices associated with the same clause form a clique.

Let $x$ be a assignment for $F$ satisfying $\kappa$ clauses. We now give an independent set in $G_F$ of size $\kappa$. For every satisfied clause $w$ and and $\ell$ the assignment to variables of $w$ in $x$ we choose the vertex $w, \ell$ in the independent set. As there are $\kappa$ satisfied clauses we choose $\kappa$ vertices. These vertices form and independent set as if two of these vertices $u, v$ had an edge between them it would mean that the assignments $T_u$ and $T_v$ contradict each other. This is not possible as all these assignments are partial assignments of $x$.

Consider $S$ to be an independent set in $G_F$ of size $\kappa$. We now give an assignment to $F$ which satisfies $\kappa$ clauses. Note that from vertices corresponding to the same clauses only 1 vertex can be a part of independent set as they all form a clique. Hence vertices associated with $\kappa$ different clauses must be part of the independent set. For a vertex $u$ associated with $w, \ell$ the partial assignment $T_u$ satisfies $w$. For two vertices $u, v$ in the independent set the partial assignments from $T_v$ and $T_u$ do not contradict as otherwise there would be an edge between $u$ and $v$. Hence we can join all the partial assignments $T_v$ for vertices $v$ in the independent set to get a partial assignment which satisfies $\kappa$ clauses in $F(\tau)$.  Hence the maximum independent set in $G_F(\pi_\tau)$ has size at most the maximum number clauses satisfied in $F(\tau)$.
\end{proof}

We next present the proof of Lemma~\ref{lem:comb} which just follows by combining Lemma~\ref{lem:ecc+pcpp},~\ref{lem:ser-rep-par}, and~\ref{lem:fglss-par} sequentially.

\begin{proofof}{Lemma~\ref{lem:comb}}
The proof follows by applying Lemma~\ref{lem:ecc+pcpp},~\ref{lem:ser-rep-par} and~\ref{lem:fglss-par} sequentially.

We start from a circuit $D$ with input variables $X$ ($\abs{X} = n$) and size $m > n$. Lemma~\ref{lem:ecc+pcpp} transform this into a 3-SAT instance $F$ with $\poly(m)$ clauses on the variable set $Y \cup Z$, where $\abs{Y} \leq \poly(n)$, $\abs{Z} \leq \poly(m)$, and the following hold for all $x \in \{0,1\}^n$:

\begin{itemize}
    \item If $D(x) = 0$ then $F(Y = \ENC(x))$ on variable set $Z$ has a satisfying assignment $z_x$. Furthermore, there is a $\poly(m)$-time algorithm that given $x$ outputs $z_x$.
    \item if $D(x) = 1$ then there is no assignment to the $Z$ variables in $F(Y = \ENC(x))$ satisfying more than a $(1-\Omega(1))$-fraction of the clauses.
\end{itemize}

where $\ENC: \{0, 1\}^n \to \{0, 1\}^{O(n)}$ is the linear encoding function from Theorem~\ref{thm:code}.

Applying Lemma~\ref{lem:ser-rep-par} on $F$ gives us a $O(k)$-SAT formula $F'$ on the same $Y \cup Z$ variables with $\poly(m) \cdot 2^{O(k)}$ clauses such that:
\begin{enumerate}
    \item If $Y, Z = y, z$ satisfies $F$ then $y, z$ satisfies $F'$.
    \item If $F(Y = y)$ is at most $1-\Omega(1)$ satisfiable then $F'(Y = y)$ is at most $1/2^k$ satisfiable.
\end{enumerate}

which implies that:
\begin{itemize}
    \item If $D(x) = 0$ then $F'(Y = \ENC(x))$ on variable set $Z$ has a satisfying assignment $z_x$. Furthermore, there is a $\poly(m)$-time algorithm that given $x$ outputs $z_x$.
    \item if $D(x) = 1$ then there is no assignment to the $Z$ variables in $F'(Y = \ENC(x))$ satisfying more than a $1/2^k$-fraction of the clauses.
\end{itemize}

Finally applying Lemma~\ref{lem:fglss-par} to $F'$ where we consider partial assignments $\tau$ which assign $Y$ to $\ENC(x)$ for some $x$. Hence $\tau(Y_i) = \ENC_i(x)$. As $\tau$ is fixed by fixing $x$ we rename $\pi_{\tau}$ to $\pi_{x}$. $S_{v}$ is just a renaming of $T_{v}$. Size of the graph is $\poly(n+m, \poly(m) \cdot 2^{O(k)}, 2^{O(k)}) = \poly(m, 2^k)$ as $m > n$.
\end{proofof}

\section{Main Result}\label{sec:main}
We now turn to the proof of the main result, Theorem~\ref{thm:main}. We will prove the result for $\EMAJ \circ \mathcal{C}$ first, and sketch how to extend to $f \circ \mathcal{C}$ for sparse symmetric $f$ in Section~\ref{subsec:sparse}. Below we prove $\EMAJ \circ \mathcal{C}$ lower bounds for $\qNP$ when we have $2^{n-n^{\eps}}$ time algorithms for \#SAT on $\mathcal{C}$ circuits of size $2^{n^{\eps}}$. For the other parts of Theorem~\ref{thm:main} (on \#SAT algorithms with running time $2^{n}/n^{\omega(1)}$), see Section~\ref{sec:np-nexp}.

We note here that in Theorem~\ref{thm:main} we mentioned polynomial size lower bounds for $\EMAJ \circ \mathcal{C}$ we in fact prove quasi-polynomial size lower bounds below.

\begin{theorem}\label{thm:nqp-res}
Suppose $\mathcal{C}$ is typical, and the parity function has $\poly(n)$-sized $\mathcal{C}$ circuits. Then for every $k$, $\NQP$ does not have $\EMAJ \circ \mathcal{C} = \mathcal{H}$ circuits of size $O(n^{\log^k n})$, if for some $\epsilon \in (0, 1)$ there is a \#SAT algorithm running in time $2^{n-n^{\eps}}$ for all circuits from class $\mathcal{C}$ of size at most $2^{n^{\eps}}$.
\end{theorem}

\begin{proof}
Let us assume that for a fixed $k > 0$, $\NQP$ has $\mathcal{H} = \EMAJ \circ \mathcal{C}$ circuits of size $O(n^{(\log^k n)})$ which implies that $\NQP \in \size(n^{O(\log^k n)})$ for general circuits. By Theorem~\ref{lem:nqp-lb-CAPP2}, we obtain a contradiction if for some constant $\delta \in (0, 1)$ and $g(n) = 2^{n^{2\delta}}$ we can give a $2^{n}/g(n)^{\omega(1)}$ time nondeterministic algorithm for distinguishing between:
\begin{enumerate}
    \item YES case: $D$ has no satisfying assignments.
    \item NO case: $D$ has at least $2^{n}\left(1-1/g(n)\right)$ satisfying assignments
\end{enumerate}
given a generic fan-in 2 circuit $D$ with $n$ inputs and size $m \leq h(n) :=  2^{n^{\delta}}$. Under the hypothesis, we will give such an algorithm for $\delta = \eps/4$.

Using Lemma~\ref{lem:comb}, we reduce the circuit $D$ to an independent set instance $G_D$ (with $k = \log{h(n)}$) on $n_2 =  \poly(m, 2^{O(k)}) = \poly(m, 2^{O(k)}) = \poly(m, h(n)^{O(1)}) = \poly(h(n))$ vertices. We also find subsets $S_i$ for every vertex $i \in [n_2]$. Let $\pi_x$ be the partial assignment which assigns a vertex $i$ to $0$ if there exist $(j', b) \in S_i$ such that $\ENC_{j'}(x) \neq b$. Note that $\pi_x$ does not assign any vertex to $1$. By Lemma~\ref{lem:comb}, $G_D$ has the following properties:
\begin{enumerate}
    \item If $D(x) = 0$, then $G_D(\pi_x)$ has an independent set of size $\kappa$. Furthermore, given $x$ we can find this independent set in $\poly(h(n))$ time.
    \item If $D_1(x) = 1$, then in $G_D(\pi_x)$, all independent sets have size at most $\kappa/h(n)$.
\end{enumerate}

This means it suffices for us to distinguish between the following two cases:
\begin{enumerate}
    \item YES case: For all $x$, $G_D(\pi_x)$ has an independent set of size $\kappa$.
    \item NO case: For at most $2^{n}/g(n)$ values of $x, G_D(\pi_x)$ has an independent set of size $\geq \kappa/h(n)$.
\end{enumerate}

\textbf{Guessing a succinct witness circuit: }As guaranteed by Lemma~\ref{lem:comb} given an $x$ such that $D(x) = 0$ we can find the assignment $A(x)$ to $G_D$ which is consistent with $\pi_x$ and represents an independent set of size $\kappa$ in $\poly(h(n))$ time. Let $A(x, i)$ denote the assignment to the $i^{th}$ vertex in $A(x)$. Given $x$ and vertex $i \in [n_2]$, in time $\poly(h(n))$ we can produce $\neg A(x, i)$.
\begin{claim}\label{clm:small-ckt}
Under the hypothesis, there is a $h(n)^{o(1)}$-sized $\EMAJ \circ \mathcal{C}$ circuit $U$ of size $h(n)^{o(1)}$ with $x, i$ as input representing $\neg A(x, i)$. 
\end{claim}
\begin{proof}
Under the hypothesis, for some constant $k$, we have $\NQP \subseteq \size_{\mathcal{H}}[n^{\log^k n}]$. Specifically, for $p(n) = n^{\log^{k+1} n}$ we have $\NTIME[p(n)] \subseteq \size_{\mathcal{H}}[p(n)^{1/\log n}] \subseteq \size_{\mathcal{H}}[p(n)^{o(1)}]$. As $h(n) = 2^{n^{\eps}} \gg p(n)$, a standard padding argument implies $\NTIME[\poly(h(n))] \subseteq \size_{\mathcal{H}}[(\poly(h(n)))^{o(1)}] = \size_{\mathcal{H}}[h(n)^{o(1)}]$. Since $\neg A(x, i)$ is computable in $\poly(h(n))$ time, we have that $\neg A(x, i)$ can be represented by a $h(n)^{o(1)}$-sized $\mathcal{H} = \EMAJ \circ \mathcal{C}$ circuit.
\end{proof}

Our nondeterministic algorithm for GAP-UNSAT begins by guessing $U$ guaranteed by Claim~\ref{clm:small-ckt} which is supposed to represent $\neg A$. Then by the reduction in Lemma~\ref{lem:emaj-to-psum} we can covert $U$ to a $\PSUM \circ \mathcal{C}$ circuit $R$ for $A(x, i)$ of size $\poly(h(n)^{o(1)}) = h(n)^{o(1)}$. Note that if our guess for $U$ is correct, i.e., $U = \neg A$, then $R$ represents $A$.

Let the subcircuits of $R$ be $R_1, R_2, \ldots, R_{t}$, so that $R(x) = \sum_{j \in [t]} R_{j}$, where $R_{j} \in \mathcal{C}$ and $t \leq h(n)^{o(1)}$. The number of inputs to $R_{j}$ is $n' = \abs{x}+\log{n_2}=n+O(\log h(n))$, and the size of $R_{j}$ is $h(n)^{o(1)}$.

Note that $R(x, i) = 0$ represents that the $i^{th}$ vertex is not in the independent set of $G_D$ in a solution corresponding to $x$, while $R(x, i) > 0$ represents that it is in the independent set of $G_D$ in a solution corresponding to $x$. For all $x$ and $i$ we have $0 \leq R(x, i) \leq t \leq h(n)^{o(1)}$.\\

\textbf{Verifying that $R$ encodes valid independent sets:} We can verify that the circuit $R$ produces an independent set on all $x$ by checking each edge over all $x$. To check the edge between vertices $i_1$ and $i_2$ we need to verify that at most one of them is in the independent set. Equivalently, for all $x$ we check that $R(x, i_1) \cdot R(x, i_2) = 0$. As $R(x, i) \geq 0$ for all $x$ and $i$ we can just verify $$\sum_{x \in \{0, 1\}^n} R(x, i_1) \cdot R(x, i_2) = 0.$$ Since $R(x, i) = \sum_{j \in [t]} R_j(x, i)$ it suffices to verify that
$$\sum_{x \in \{0, 1\}^n} \sum_{j_1, j_2 \in [t]} R_{j_1}(x, i_1) \cdot R_{j_2}(x, i_2) = 0.$$ Let $R_{j_1, j_2}(x, i_1, i_2) = R_{j_1}(x, i_1) \cdot R_{j_2}(x, i_2)$. Since $\mathcal{C}$ is closed under AND (upto polynomial factors) $R_{j_1, j_2}$ also has a $\poly(h(n)^{o(1)}) = h(n)^{o(1)}$ sized $\mathcal{C}$ circuit. Exchanging the order of summations is suffices for us to verify
$$\sum_{j_1, j_2 \in [t]} \left(\sum_{x \in \{0, 1\}^n} R_{j_1, j_2}(x, i_1, i_2)\right) = 0.$$
For fixed $i_1, i_2, j_1, j_2$ the number of inputs to $R_{j_1, j_2}$ is $\abs{x}=n$ and its size is $h(n)^{o(1)} \leq 2^{n^{\eps}}$. Hence, for fixed $i_1, i_2, j_1, j_2$ we can compute $\sum_x R_{j_1, j_2}(x, i_1, i_2)$ using the \#SAT algorithm from our assumption, in time $2^{n-n^{\eps}}$. Summing over all $j_1, j_2$ pairs only adds another multiplicative factor of $t^2 = h(n)^{o(1)}$. This allows us to verify that the edge $(i_1, i_2)$ is satisfied by $R$. Checking all edges of $G_D$ only adds another multiplicative factor of $\poly(h(n))$. Hence the total running time for verifying that $R$ encodes valid independent sets on all $x$ is still $2^{n-n^{\eps}}\poly(h(n))$.\\

\textbf{Verifying consistency of independent set produced by $R$ with $\pi_x$: } As we care about the sizes of independent sets in $G_D(\pi_x)$ over all $x$ we need to check if the assignment by $R$ is consistent with $\pi_x$. As $\pi_x$ only assigns vertices to $0$, we need to verify that all vertices assigned to $0$ in $\pi_x$ are in fact assigned to $0$ by the assignment given by $R(x, \cdot)$. From Lemma~\ref{lem:comb}, we know that $\pi_x$ assigns a vertex $i$ to $0$ if for some $(j', b) \in S_i$, $\ENC_{j'}(x) \neq b$. To check this condition we need to verify that $R(x, i) = 0$ if for some $(j', b) \in S_i$, $\ENC_{j'}(x) \neq b$. Equivalently, we cn check $(\ENC_{j'}(x) \oplus b) \cdot R(x, i) = 0$ for all $x, i, (j', b) \in S_i$. Since $(\ENC_j(x) \oplus b)R(x, i) \geq 0$ for all possible inputs we can just check that $$\sum_{x \in \{0, 1\}^n} (\ENC_{j'}(x) \oplus b) \cdot R(x, i) = 0$$ for all $i, (j', b) \in S_i$. As $R(x, i) = \sum_{j \in [t]} R_j(x, i)$ we can equivalently verify that  $$\sum_{x \in \{0, 1\}^n} \sum_{j \in [t]} (\ENC_{j'}(x) \oplus b) \cdot R_{j}(x, i) =  0$$ for all $i, (j', b) \in S_i$. Note that $R_{j'}(x, i)$ has a $h(n)^{o(1)}$ sized $\mathcal{C}$ circuit. By our assumption parity has a $\poly(n)$-sized $\mathcal{C}$-circuit so $(\ENC_j(x) \oplus b)$ also has a $\poly(n)$-sized $\mathcal{C}$ circuit. Hence $(\ENC_j(x) \oplus b) \cdot R_{j'}(x, i)$ has a $\poly(n, h(n)^{o(1)}) = h(n)^{o(1)}$-sized $\mathcal{C}$ circuit, since $\mathcal{C}$ is closed under AND. 

For fixed $(i, j, j')$, $(\ENC_{j'}(x) \oplus b) \cdot R_{j}(x, i) \in \mathcal{C}$ has $\abs{x} = n$ inputs and size $h(n)^{o(1)} < 2^{n^{\eps}}$. Hence we can use our assumed \#SAT algorithm to calculate $\sum_{x \in \{0, 1\}^n} (\ENC_{j'}(x) \oplus b) \cdot R_{j}(x, i)$ in time $2^{n-n^{\eps}}$. Summing over all $j \in [t]$ introduces another multiplicative factor of $h(n)^{o(1)}$. This allows us to verify the desired condition for a fixed $i, (j', b) \in S_i$. To check it for all $i, (j', b) \in S_i$ (recall $\abs{S_i} = O(n)$ by Theorem~\ref{thm:code}) only introduces another multiplicative factor of $\poly(h(n)) \cdot O(n) = \poly(h(n))$ in time. Therefore the total running time for verifying consistency w.r.t. $\pi_x$ is $2^{n-n^{\eps}}\poly(h(n))$.

At this point, we now know that $R$ represents an independent set, and that $R$ is consistent with $\pi_x$. We need to distinguish between:
\begin{enumerate}
    \item YES case: For all $x$, $R(x, \cdot)$ represents an independent set of size $\kappa$.
    \item NO case: For at most $2^{n}/g(n)$ values of $x, R(x, \cdot)$ represents an independent set of size $\geq \kappa/h(n)$.
\end{enumerate}



\begin{lemma}\label{lem:approx}
For all $x$ such that $R(x, \cdot)$ represents an independent set of size $a$. we have $a \leq \sum_{i \in [n_2]} R(x, i) \leq at$.
\end{lemma}
\begin{proof}
For every vertex $i$ in the independent set, $1 \leq R(x, i) \leq t$. For all vertices $i$ not in the independent set, we have $R(x, i) = 0$.  Hence $a \leq \sum_{i \in [n_2]} R(x, i) \leq at$.
\end{proof}

\textbf{Distinguishing between the YES and NO cases:}
To distinguish between the YES and NO cases, we now compute 
\begin{align}\label{last}
&\sum_{x \in \{0, 1\}^n} \sum_{i \in [n_2]} R(x, i)
\end{align}This allows us to distinguish between the YES case and NO case as:
\begin{enumerate}
    \item YES case: We have for at least $2^{n}(1-1/g(n))$ values of $x$ we have an independent set of size at most $\kappa/h(n)$. By Lemma~\ref{lem:approx} for such $x$, $\sum_{i \in [n_2]} R(x, i) \leq t\kappa/h(n)$. for the rest of $2^{n}/g(n)$ values of $x$ the independent set could be all the vertices in the graph $G_D$. Hence by Lemma~\ref{lem:approx} for such values of $x$, $\sum_{i \in [n_2]} R(x, i) \leq tn_2 = \poly(h(n))$. Hence
    \begin{align*}
        \sum_{x \in \{0, 1\}^n} \sum_{i \in [n_2]} R(x, i) &\leq (2^{n}/g(n))\poly(h(n))+ 2^nt\kappa/h(n)\\
        &\leq o(2^n)+2^nt\kappa/h(n) \hspace{15pt} [\text{As } h(n) = g(n)^{o(1)}]\\
        &\leq o(2^n)+o(2^n\kappa) \hspace{15pt} [\text{As } t = h(n)^{o(1)}]\\
        &\leq 2^n \kappa \hspace{15pt} [\text{As } \kappa > 1]
    \end{align*}

    \item NO case: We have for all $x \in \{0, 1\}^n$ the independent set is at least of size $\kappa$. Hence by Lemma~\ref{lem:approx} the sum is $\sum_{x \in \{0, 1\}^n} \sum_{i \in [n_2]} R(x, i) > 2^{n}\kappa$.
\end{enumerate}

All that remains is how to compute \eqref{last}. As $R(x, i) = \sum_{j \in [t]} R_j(x, i)$, we can compute
$$\sum_{x \in \{0, 1\}^n} \sum_{i \in [n_2]} \sum_{j \in [t]} R_{j}(x, i) = \sum_{j \in [t]} \sum_{i \in [n_2]} \sum_{x \in \{0, 1\}^n} R_{j}(x, i)$$

For a fixed $i, j$, $R_{j}(x, i) \in \mathcal{C}$, it has $\abs{x} = n$ inputs and size $\leq \poly(h(n)^{o(1)}) = h(n)^{o(1)} < 2^{n^{\eps}}$. Hence we can use the assumed \#SAT algorithm to calculate $\sum_{x \in \{0, 1\}^n} R_{j}(x, i)$ in time $2^{n-n^{\eps}}$. Summing over all $j \in [t], i \in [n_2]$ only introduces another $h(n)^{o(1)}\poly(h(n)) = \poly(h(n))$ multiplicative factor. Thus the running time for distinguishing the two cases is $2^{n-n^{\eps}}\poly(h(n))$.

In total our running time comes to $2^{n-n^{\eps}}\poly(h(n)) = 2^{n-n^{4\delta}+O(n^{\delta})} \leq 2^{n-n^{3\delta}} = 2^{n}/g(n)^{\omega(1)}$ as $g(n) = 2^{n^{2\delta}}$ and $\eps = 4\delta$. By Theorem~\ref{lem:nqp-lb-CAPP2}, this gives us a contradiction which completes our proof.

\end{proof}

The above theorem when combined with known \#SAT algorithms for $\ACC^0 \circ \THR$ gives an $\NQP$ lower bound for $\EMAJ \circ \ACC^0 \circ \THR$.

\subsection{$\EMAJ \circ \ACC^0 \circ \THR$ Lower bound}\label{subsec:acc}

We will apply a known $\#$SAT algorithm for $\ACC \circ \THR$ circuits.

\begin{theorem}[\cite{acc-algo}]\label{thm:acc_algo}
For every pair of constants $d, m$, there exists a constant $\eps \in (0, 1)$ such that \#SAT can be solved in time $2^{n-n^{\eps}}$ time for $\AC^0[m] \circ \THR$ circuits of depth $d$ and size $2^{n^{\eps}}$.
\end{theorem}

\begin{theorem}\label{thm:acc_lb}
For constants $k, d, m$, $\NQP$ does not have $\size(n^{\log^k n})$ $\EMAJ \circ \ACC^0 \circ \THR$ circuits of depth $d$.
\end{theorem}
\begin{proof}
We first note that $\ACC^0 \circ \THR$ is indeed typical and can represent $\ENC(x)$ by $\poly(n)$-sized circuits as $\ENC(x): \{0, 1\}^n \to \{0, 1\}^{O(n)}$ is a linear function.

By Theorem~\ref{thm:acc_algo} we know that for all constants $d$ there exists some constant $\epsilon \in (0, 1)$ such that there exists a \#SAT algorithm running in time $2^{n-n^{\eps}}$ for all circuits from class $\ACC^0 \circ \THR$ of size $\leq 2^{n^{\eps}}$ and depth $d$.

The above properties imply that $\ACC^0 \circ \THR$ satisfies the preconditions of Theorem~\ref{thm:nqp-res} and hence for every pair of constant $k, d$, $\NQP$ does not have $\size(n^{\log^k n})$ $\EMAJ \circ \ACC^0 \circ \THR$ circuits of depth $d$.
\end{proof}

The above theorem can be rewritten as: For constants $k, d, m$, there exists a constant $e$ such that $\NTIME[n^{\log^e n}]$ does not have $n^{\log^k n}$-size $\EMAJ \circ \ACC^0 \circ \THR$ circuits of depth $d$. Here the constant $e$ depends on $d$ and $m$. Using a standard trick (as in ~\cite{MurrayW18}) this dependence can be removed as we show below.

\begin{corollary}
There exists an $e$ such that $\NTIME[n^{\log^e n}]$ does not have polynomial size $\EMAJ \circ \ACC^0 \circ \THR$ circuits.
\end{corollary}

\begin{proof}
Assume for contradiction that for all $e$, there exists constants $d, m$ such that $\NTIME[n^{\log^e n}]$ has poly-sized $\EMAJ \circ \AC^0[m] \circ \THR$ circuit of depth $d$. This implies that $P$ has poly-sized $\EMAJ \circ \AC^0[m] \circ \THR$ circuits, which further implies that \emph{CIRCUIT EVALUATION} problem has poly-sized $\EMAJ \circ \AC^0[m_0] \circ \THR$ circuit of a fixed constant depth $d_0$ and fixed constant $m_0$. Hence any circuit of size $s$ has an equivalent $\poly(s)$-sized $\EMAJ \circ \AC^0[m_0] \circ \THR$ circuit of depth $d_0$. Combining this with our assumption yields: For all $e$, there exists constants $d, m$ such that $\NTIME[n^{\log^e n}]$ has poly-sized $\EMAJ \circ \AC^0[m_0] \circ \THR$ circuit of depth $d_0$. This contradicts Theorem~\ref{thm:acc_lb} and hence our assumption was wrong, which completes the proof.
\end{proof}

\section{Extension to All Sparse Symmetric Functions}\label{subsec:sparse}

Our lower bounds extend to circuit classes of the form $f \circ \mathcal{C}$ where $f$ denotes a family of symmetric functions that only take the value $1$ on a small number of slices of the hypercube. Formally, let $f : \{0,1\}^n \rightarrow \{0,1\}$ be a symmetric function, and let $g : \{0,1,\ldots,n\} \rightarrow \{0,1\}$ be its ``companion'' function, where for all $x$, $f(x) = g(\sum_i x_i)$ (here, $x_i$ denotes the $i$-th bit of $x$). For $k \in \{0,1,\ldots,n\}$, we say that a symmetric function $f$ is \emph{$k$-sparse} if $|g^{-1}(1)| = k$. For example, the all-zeroes function is $0$-sparse, the all-ones function is $n$-sparse, and the $\EMAJ$ function is $1$-sparse.  

\begin{theorem} Let $k < n/2$. Every $k$-sparse symmetric function $f : \{0,1\}^n \rightarrow \{0,1\}$ can be represented as an exact majority of $n^{O(k)}$ ANDs on $k$ inputs.
\end{theorem}

\begin{proof}
Given a $k$-sparse $f$ and its companion function $g$, consider the polynomial expression \[E(x) := \prod_{v \in g^{-1}(1)}\left(\sum_i x_i - v\right).\] Then $E(x) = 0$ whenever $f(x) = 1$, and $E(x) \neq 0$ otherwise. Expanding $E$ into a sum of products, we can write $E$ as a multilinear $n$-variate polynomial of degree at most $k$, with integer coefficients of magnitude at most $n^{O(k)}$ (since each $v \leq n$). We can therefore write $E$ as the EMAJORITY of $n^{O(k)}$ distinct ANDs on up to $k$ inputs. 
\end{proof}

The above theorem immediately implies that for every $k$-sparse symmetric function $f_m$, any circuit with an $f_m$ at the output gate can be rewritten as a circuit with an EMAJ of fan-in at most $m^{O(k)}$ at the output gate (and ANDs of fan-in up to $k$ below that). 

\begin{corollary} For every fixed $k$, and every $k$-sparse symmetric function family $f=\{f_n\}$, $\qNP$ does not have polynomial-size $f \circ \ACC^0 \circ \THR$ circuits.
\end{corollary}



\section{NEXP Lower Bounds}\label{sec:np-nexp}

In this section we prove NEXP Lower Bounds under weaker algorithmic assumptions. The proof follows the same pattern as the proof of lower bound for $\NQP$ in Theorem~\ref{thm:nqp-res}.

\subsection{NEXP Lower Bounds}
\begin{theorem}[\cite{Williams-jacm14}]\label{thm:nexp-lb-CAPP}
Suppose for some constant $\eps \in (0, 1)$ there is an algorithm $A$ that for all $\poly(n)$-size circuits $C$ on $n$ inputs, $A(C)$ runs in $2^{n}/n^{\omega(1)}$ time, outputs YES on all unsatisfiable $C$, and outputs NO on all $C$ that have at least $2^{n-1}$ satisfying assignments. Then $\NTIME[2^{n}] \not\subset \Ppoly$.
\end{theorem}

\begin{theorem}\label{lem:nexp-lb-CAPP2}
Suppose there is an algorithm $A$ that for all $\poly(n)$-sized circuits $C$ on $n$ inputs, $A(C)$ runs in $2^{n}/g(n)^{\omega(1)}$ time, outputs YES on all unsatisfiable $C$, and outputs NO on all $C$ that have at least $2^{n}(1-1/g(n))$ satisfying assignments, for any $g(n)$ satisfying $g(n) = n^{\omega(1)}, g(n) = 2^{o(n)}$. Then for all $k$, there is a $c \geq 1$ such that $\NTIME[2^{n}] \not\subset \Ppoly$.
\end{theorem}
\begin{proof} Our starting point is Theorem~\ref{thm:nexp-lb-CAPP} (\cite{Williams-jacm14}): we are given an $m$-input, $\poly(m)$-size circuit $D'$ that is either UNSAT or has at least $2^{m-1}$ satisfying assignments, and we wish to distinguish between the two cases with a $2^{m}/m^{\omega(1)}$-time algorithm.
 
We create a new circuit $D$ with $n$ inputs, where $n$ satisfies \[n = m + \psi \cdot \log{g(n)},\] and $\psi > 0$ is the constant from Lemma~\ref{lem:hit}. (Note that, since $g(n)$ is time constructible and $g(n) \leq 2^{o(n)}$, such an $n$ can be found in subexponential time.) Applying the algorithm from Lemma~\ref{lem:hit}, $D$ treats its $n$ bits of input as a string of randomness $r$, computes $t=O(\log g(n))$ strings $x_1, x_2, \ldots, x_t \in \{0,1\}^m$ with a $\poly(m, \log g)$-size circuit, then outputs the OR of $D'(x_i)$ over all $i=1,\ldots,t$. Note the total size of our circuit $D$ is $\poly(m,\log g) + O(\log g)\cdot \size(D') = \poly(m) = \poly(n)$.

Clearly, if $D'$ is unsatisfiable, then $D$ is also unsatisfiable. By Lemma~\ref{lem:hit}, if $D'$ has $2^{m-1}$ satisfying assignments, then $D$ has at least $2^{n}(1-1/g(n))$ satisfying assignments. As $\size(D) \leq \poly(n)$, by our assumption we can distinguish the case where $D$ is unsatisfiable from the case where $D$ has at least $2^{n}(1-1/g(n))$ satisfying assignments, with an algorithm running in time $2^{n}/g(n)^{\omega(1)}$. This yields an algorithm for distinguishing the original circuit $D'$ on $m$ inputs and $\poly(m)$ size, running in time \[2^{n}/g(n)^{\omega(1)}
= 2^{m}g(n)^{O(1)}/g(n)^{\omega(1)} = 2^{m}/g(n)^{\omega(1)} \leq 2^{m}/g(m)^{\omega(1)} \leq 2^{m}/m^{\omega(1)}\] since $n > m, g(n) = n^{\omega(1)}$. By Theorem~\ref{thm:nexp-lb-CAPP}, this implies that $\NTIME[2^{n}] \not\subset \Ppoly$
\end{proof}

\begin{theorem}\label{thm:nexp-res}
$NTIME[2^n]$ does not have $\poly(n)$-sized $\EMAJ \circ \mathcal{C} = \mathcal{H}$ circuits if
\begin{enumerate}
    \item There exists a \#SAT algorithm running in time $2^{n}/b(n)$ for all $\poly(n)$-sized circuits from class $\mathcal{C}$ where $b(n) = n^{\omega(1)}$
    \item $\mathcal{C}$ is {\emph typical} and $(\neg) \ENC_i(x)$ has $\poly(n)$-sized $\mathcal{C}$ circuits.
\end{enumerate}
\end{theorem}

\begin{proof}
Let us assume that $\NTIME[2^n]$ has $\poly(n)$-sized $\mathcal{H} = \EMAJ \circ \mathcal{C}$ circuits which implies that $\NTIME[2^n] \in \Ppoly$. By Theorem~\ref{lem:nexp-lb-CAPP2}, we will get a contradiction if we can give a $2^{n}/g(n)^{\omega(1)}$ time nondeterministic algorithm for distinguishing between:
\begin{enumerate}
    \item YES case: $D$ has no solutions.
    \item NO case: $D$ has at least $2^{n}\left(1-1/g(n)\right)$ solutions.
\end{enumerate}
given a circuit $D$ with $n$ inputs and size $m = \poly(n)$ where $g(n) = n^{\omega(1)}$. We will take a $g(n)$ such that $g(n) = b(n)^{o(1)}$.

Let $h(n)$ be a function such that $h(n) = g(n)^{o(1)}, h(n) = n^{\omega(1)}$. Using Lemma~\ref{lem:comb} we reduce $D$ to independent set instance on $G_D$ (with $k = \log{h(n)}$) over $n_2 =  \poly(m, 2^{O(k)}) = \poly(m, h(n)) = \poly(h(n))$ vertices and edges as $h(n) = n^{\omega(1)}$ and $m = \poly(n)$. We also find $S_i$ for every vertex $i \in [n_2]$. By Lemma~\ref{lem:comb}, $G_D$ has the following properties:
\begin{enumerate}
    \item Let $D(x) = 0$ then for $G_D(\pi_x)$ there exists an independent set of size $\kappa$. Further given $x$ we can find this assignment in $\poly(h(n))$ time.
    \item Let $D_1(x) = 1$ then for $G_D(\pi_x)$ all independent sets have size $\leq \kappa/h(n)$.
\end{enumerate}
where $\pi_x$ is the partial assignment which assigns a vertex $i$ to $0$ if there exist $(j', b) \in S_i$ such that $\ENC_{j'}(x) \neq b$. $\pi_x$ does not assign any vertex to $1$.

This means we need to distinguish between the following two cases:
\begin{enumerate}
    \item YES case: For all $x$, $G_D(\pi_x)$ has an independent set of size $\kappa$.
    \item NO case: For at most $2^{n}/g(n)$ values of $x, G_D(\pi_x)$ has an independent set of size $\geq \kappa/h(n)$.
\end{enumerate}

\textbf{Guessing a succinct witness circuit: }As given an $x$ such that $D(x) = 1$ we can find the assignment $A(x)$ to $G_D$ which is consistent with $\pi_x$ and represents an independent set of size $\kappa$ in $\poly(h(n))$ time. Let $A(x, i)$ denote the assignment to $i^{th}$ vertex in $A(x)$. Given $x$ and vertex $i \in [n_2]$ in time $\poly(h(n))$ we can produce $\neg A(x, i)$.
\begin{claim}\label{clm:small-ckt-nexp}
There exists a $\poly(n)$-sized $\EMAJ \circ \mathcal{C}$ circuit $U$ with $x, i$ as input representing $A(x, i)$. 
\end{claim}
\begin{proof}
 As given $x$ and vertex $i \in [n_2]$ in time $\poly(h(n))$ we can produce $\neg A(x, i)$. $\NTIME[2^n]$ has poly-sized $\EMAJ.\mathcal{C}$ circuits given $x$ and $i \in [n_2]$ we can also produce/represent $\neg A(x)_i$ by a $\poly(n+O(\log h(n))) = \poly(n)$ $\EMAJ.\mathcal{C}$ circuit.

\end{proof}

Our nondeterministic algorithm for GAP-UNSAT begins by guessing $U$ guaranteed by Claim~\ref{clm:small-ckt-nexp} which is supposed to represent $\neg A$. Then by the reduction in Lemma~\ref{lem:emaj-to-psum} we can covert $U$ to a $\PSUM \circ \mathcal{C}$ circuit $R$ for $A(x, i)$ of size $\poly(n)$. Note that if our guess for $U$ is correct i.e. $U = \neg A$ then $R$ represents $A$.

Let subcircuits of $R$ be $R_1, R_2, \ldots, R_{t}$ i.e. $R(x) = \sum_{j \in [t]} R_{j}$ where $R_{j} \in \mathcal{C}$ and $t = \poly(n)$. The number of inputs to $R_{j}$ are $n' = \abs{x}+\log{n_2}=n+O(\log h(n))$ and the size of $R_{j}$ is $\poly(n)$.

Note that $R(x, i) = 0$ represents that the $i^{th}$ vertex is not part of the independent set in a solution corresponding to $x$ while $R(x, i) > 0$ represents that it is part of the independent set in a solution corresponding to $x$. For all $x, i$, $0 \leq R(x, i) \leq t \leq \poly(n)$.\\

\textbf{Verifying that $R$ encodes valid independent sets:} We can verify that the circuit produces an independent set by checking each edge over all $x$. To check the edge between vertices $i_1$ and $i_2$ we need to verify that most one of them is part of the independent set. Equivalently, for all $x$, $R(x, i_1) \cdot R(x, i_2) = 0$. As $R(x, i)$ is always $\geq 0$ we can just verify $$\sum_x R(x, i_1) \cdot R(x, i_2) = 0.$$ Since $R(x, i) = \sum_{j \in [t]} R_j(x, i)$ it suffices to verify that
$$\sum_x \sum_{j_1, j_2 \in [t]} R_{j_1}(x, i_1) \cdot R_{j_2}(x, i_2) = 0$$ Let $R_{j_1, j_2}(x, i_1, i_2) = R_{j_1}(x, i_1) \cdot R_{j_2}(x, i_2)$. By definition~\ref{def:typical}, $\mathcal{C}\cdot\mathcal{C} = \AND_2.\mathcal{C} = \mathcal{C}$ we $R_{j_1, j_2}$ has a $\poly(n)$ sized $\mathcal{C}$ circuit. Interchanging the summations we get that we need to verify
$$\sum_{j_1, j_2 \in [t]} \sum_x R_{j_1, j_2}(x, i_1, i_2) = 0$$
For a fixed $j_1, j_2$ number of inputs to $R_{j_1, j_2}$ are $\abs{x}=n$ and its size is $\poly(n)$. Hence, for a fixed pair of $j_1, j_2$ we can compute $\sum_x R_{j_1, j_2}(x, i_1, i_2)$ using the \#SAT algorithm from our assumption in time $2^{n}/b(n)$. Going over all $j_1, j_2$ pairs only adds another multiplicative factor of $t^2 = \poly(n)$. This allows us to verify that the edge $(i_1, i_2)$ is satisfied by $R$. 

Checking all edges only adds another multiplicative factor of $\poly(h(n))$. Hence the total running time for verifying that $R$ encodes valid independent sets is still $2^{n}\poly(h(n))/b(n)$.\\

\textbf{Verifying consistency of independent set produced by $R$ with $\pi_x$: } As we care about the size of independent set in $G_D(\pi_x)$ while $R$ assigns all vertices in $G_D$ we need to check if the assignment by $R$ is consistent with $\pi_x$. As $\pi_x$ only assigns vertices to $0$ we need to verify that all vertices assigned to $0$ in $\pi_x$ are in fact assigned to $0$ by the assignment given by $R(x, \cdot)$. From Lemma~\ref{lem:comb} we know that $\pi_x$ assigns a vertex $i$ to $0$ if for any $(j', b) \in S_i$, $\ENC_{j'}(x) \neq b$. To check this we need to verify that $R(x, i) = 0$ whenever for any $(j', b) \in S_i$, $\ENC_{j'}(x) \neq b$. Equivalently, $(\ENC_{j'}(x) \oplus b)R(x, i) = 0$ for all $x, i, (j', b) \in S_i$. As $(\ENC_j(x) \oplus b)R(x, i) \geq 0$ we can just check that $$\sum_x (\ENC_{j'}(x) \oplus b)R(x, i) = 0$$ for all $i, (j', b) \in S_i$. As $R(x, i) = \sum_{j \in [t]} R_j(x, i)$ we can equivalently verify that  $$\sum_{x \in \{0, 1\}^n} \sum_{j \in [t]} (\ENC_{j'}(x) \oplus b)R_{j}(x, i) =  0$$ for all $i, (j', b) \in S_i$. Note that $R_{j'}(x, i)$ has a $\poly(n)$-sized $\mathcal{C}$ circuit. By our assumption $(\ENC_j(x) \oplus b)$ has a $\poly(n)$-sized $\mathcal{C}$ circuit. Hence $(\ENC_j(x) \oplus b)R_{j'}(x, i)$ has a $\poly(n)$-sized $\mathcal{C}$ circuit as we are given that $\mathcal{C}$ is typical. 

For fixed $(i, j, j')$, $(\ENC_{j'}(x) \oplus b)R_{j}(x, i) \in \mathcal{C}$ has $\abs{x} = n$ inputs and size $\poly(n)$. Hence we can use \#SAT algorithm from assumption to calculate $\sum_{x \in \{0, 1\}^n} (\ENC_{j'}(x) \oplus b)R_{j}(x, i)$ in time $2^{n}/b(n)$. Going over all $j \in [t]$ adds another multiplicative factor of $\poly(n)$. This allows us to verify the condition for a fixed $i, (j', b) \in S_i$.

To go all $i, (j', b) \in S_i$ ($\abs{S_i} = O(n)$ by Theorem~\ref{thm:code}) only adds another multiplicative factor of $\poly(h(n)) \cdot O(n) = \poly(h(n))$ in time. The total running time for verifying consistency w.r.t. $\pi_x$ is $2^{n}\poly(h(n))/b(n)$.

As we now know that $R$ represents and independent set and that $R$ is consistent with $\pi_x$ we need to distinguish between:
\begin{enumerate}
    \item YES case: For all $x$, $R(x, \cdot)$ represents an independent set of size $\kappa$.
    \item NO case: For at most $2^{n}/g(n)$ values of $x, R(x, \cdot)$ represents an independent set of size $\geq \kappa/h(n)$.
\end{enumerate}
This is because we are giving a non-deterministic algorithm, and hence we can assume in the YES case that $R = A$.\\



\begin{claim}\label{lem:approx-nexp}
For an $x$ such that $R(x, \cdot)$ represents an independent set of size $a$ then $a \leq \sum_{i \in [n_2]} R(x, i) \leq at$.
\end{claim}
\begin{proof}
For every vertex $i$ which is part of the independent set we have $1 \leq R(x, i) \leq t$ while for all vertices $i$ which are not part of the independent set we have $R(x, i) = 0$.  Hence $b \leq \sum_{i \in [n_2]} R(x, i) \leq bt$.
\end{proof}

\textbf{Distinguishing between YES and NO cases: }
To distinguish between YES and NO cases we compute $$\sum_{x \in \{0, 1\}^n} \sum_{i \in [n_2]} R(x, i)$$ This allows us to distinguish between the YES case and NO case as:
\begin{enumerate}
    \item YES case: We have for at least $2^{n}(1-1/g(n))$ values of $x$ we have an independent set of size at most $\kappa/h(n)$. By Lemma~\ref{lem:approx-nexp} for such $x$, $\sum_{i \in [n_2]} R(x, i) \leq t\kappa/h(n)$. for the rest of $2^{n}/g(n)$ values of $x$ the independent set could be all the vertices in the graph $G_D$. Hence by Lemma~\ref{lem:approx-nexp} for such values of $x$, $\sum_{i \in [n_2]} R(x, i) \leq tn_2 = \poly(h(n))$. Hence
    \begin{align*}
        \sum_{x \in \{0, 1\}^n} \sum_{i \in [n_2]} R(x, i) &\leq (2^{n}/g(n))\poly(h(n))+ 2^nt\kappa/h(n)\\
        &\leq o(2^n)+2^nt\kappa/h(n) \hspace{15pt} [\text{As } h(n) = g(n)^{o(1)}]\\
        &\leq o(2^n)+o(2^n\kappa) \hspace{15pt} [\text{As } t = h(n)^{o(1)}]\\
        &\leq 2^n \kappa \hspace{15pt} [\text{As } \kappa > 1]
    \end{align*}

    \item NO case: We have for all $x \in \{0, 1\}^n$ the independent set is at least of size $\kappa$. Hence by Lemma~\ref{lem:approx-nexp} the sum is $\sum_{x \in \{0, 1\}^n} \sum_{i \in [n_2]} R(x, i) > 2^{n}\kappa$.
\end{enumerate}

All that remains is how to compute $\sum_{x \in \{0, 1\}^n} \sum_{i \in [n_2]} R(x, i)$. As $R(x, i) = \sum_{j \in [t]} R_j(x, i)$ we can compute
$$\sum_{x \in \{0, 1\}^n} \sum_{i \in [n_2]} \sum_{j \in [t]} R_{j}(x, i) = \sum_{j \in [t]} \sum_{i \in [n_2]} \sum_{x \in \{0, 1\}^n} R_{j}(x, i)$$

For a fixed $i, j$, $R_{j}(x, i) \in \mathcal{C}$, it has $\abs{x} = n$ inputs and size $\poly(n)$. Hence we can use \#SAT algorithm from assumption to calculate $\sum_{x \in \{0, 1\}^n} R_{j}(x, i)$ in time $2^{n}/b(n)$. Doing the summation for all $j \in [t], i \in [n_2]$ add another $h(n)^{o(1)}\poly(h(n)) = \poly(h(n))$ multiplicative factor.  The running time for distinguishing YES case and NO case is $2^{n}\poly(h(n))/b(n)$.

In total our running time comes to $2^{n}\poly(h(n))/b(n) = 2^{n}/g(n)^{\omega(1)}$. By Theorem~\ref{lem:nexp-lb-CAPP2} this gives us a contradiction which completes our proof.


\end{proof}

\bibliographystyle{alpha}
\bibliography{article}

\newcommand{\etalchar}[1]{$^{#1}$}
\begin{thebibliography}{FGL{\etalchar{+}}91}

\bibitem[AB09]{Arora-Barak}
Sanjeev Arora and Boaz Barak.
\newblock {\em Computational Complexity - {A} Modern Approach}.
\newblock Cambridge University Press, 2009.

\bibitem[ACW16]{AlmanCW16}
Josh Alman, Timothy~M. Chan, and R.~Ryan Williams.
\newblock Polynomial representations of threshold functions and algorithmic
  applications.
\newblock In {\em {FOCS}}, pages 467--476, 2016.

\bibitem[BRS95]{pp-closed}
Richard Beigel, Nick Reingold, and Daniel~A. Spielman.
\newblock {PP} is closed under intersection.
\newblock {\em J. Comput. Syst. Sci.}, 50(2):191--202, 1995.

\bibitem[BTT92]{BeigelTT92}
Richard Beigel, Jun Tarui, and Seinosuke Toda.
\newblock On probabilistic {ACC} circuits with an exact-threshold output gate.
\newblock In {\em Algorithms and Computation, Third International Symposium,
  {ISAAC} '92, Nagoya, Japan, December 16-18, 1992, Proceedings}, pages
  420--429, 1992.

\bibitem[Che19]{Chen19}
Lijie Chen.
\newblock Non-deterministic quasi-polynomial time is average-case hard for
  {ACC} circuits.
\newblock In {\em 60th {IEEE} Annual Symposium on Foundations of Computer
  Science, {FOCS} 2019, Baltimore, Maryland, USA, November 9-12, 2019}, pages
  1281--1304, 2019.

\bibitem[COS18]{ChenOS18}
Ruiwen Chen, Igor~Carboni Oliveira, and Rahul Santhanam.
\newblock An average-case lower bound against $acc^0$.
\newblock In {\em {LATIN} 2018: Theoretical Informatics - 13th Latin American
  Symposium, Buenos Aires, Argentina, April 16-19, 2018, Proceedings}, pages
  317--330, 2018.

\bibitem[CW19]{ChenW19}
Lijie Chen and R.~Ryan Williams.
\newblock Stronger connections between circuit analysis and circuit lower
  bounds, via pcps of proximity.
\newblock In {\em 34th Computational Complexity Conference, {CCC} 2019, July
  18-20, 2019, New Brunswick, NJ, {USA.}}, pages 19:1--19:43, 2019.

\bibitem[DR06]{ser-rep}
Irit Dinur and Omer Reingold.
\newblock Assignment testers: Towards a combinatorial proof of the {PCP}
  theorem.
\newblock {\em {SIAM} J. Comput.}, 36(4):975--1024, 2006.

\bibitem[FGL{\etalchar{+}}91]{FGLSS}
Uriel Feige, Shafi Goldwasser, L{\'{a}}szl{\'{o}} Lov{\'{a}}sz, Shmuel Safra,
  and Mario Szegedy.
\newblock Approximating clique is almost np-complete.
\newblock In {\em 32nd Annual Symposium on Foundations of Computer Science, San
  Juan, Puerto Rico, 1-4 October 1991}, pages 2--12, 1991.

\bibitem[Gol11]{hitting-set}
Oded Goldreich.
\newblock A sample of samplers: {A} computational perspective on sampling.
\newblock In {\em Studies in Complexity and Cryptography. Miscellanea on the
  Interplay between Randomness and Computation - In Collaboration with Lidor
  Avigad, Mihir Bellare, Zvika Brakerski, Shafi Goldwasser, Shai Halevi, Tali
  Kaufman, Leonid Levin, Noam Nisan, Dana Ron, Madhu Sudan, Luca Trevisan,
  Salil Vadhan, Avi Wigderson, David Zuckerman}, pages 302--332. 2011.

\bibitem[Gre00]{Green00}
Frederic Green.
\newblock A complex-number fourier technique for lower bounds on the mod-m
  degree.
\newblock {\em Computational Complexity}, 9(1):16--38, 2000.

\bibitem[Han07]{Hansen07}
Kristoffer~Arnsfelt Hansen.
\newblock Computing symmetric boolean functions by circuits with few exact
  threshold gates.
\newblock In {\em Computing and Combinatorics, 13th Annual International
  Conference, {COCOON} 2007, Banff, Canada, July 16-19, 2007, Proceedings},
  pages 448--458, 2007.

\bibitem[Han09]{Hansen09}
Kristoffer~Arnsfelt Hansen.
\newblock Depth reduction for circuits with a single layer of modular counting
  gates.
\newblock In {\em Computer Science - Theory and Applications, Fourth
  International Computer Science Symposium in Russia, {CSR} 2009, Novosibirsk,
  Russia, August 18-23, 2009. Proceedings}, pages 117--128, 2009.

\bibitem[HP10]{HansenP10}
Kristoffer~Arnsfelt Hansen and Vladimir~V. Podolskii.
\newblock Exact threshold circuits.
\newblock In {\em Proceedings of the 25th Annual {IEEE} Conference on
  Computational Complexity, {CCC} 2010, Cambridge, Massachusetts, USA, June
  9-12, 2010}, pages 270--279, 2010.

\bibitem[KI04]{KabanetsI04}
Valentine Kabanets and Russell Impagliazzo.
\newblock Derandomizing polynomial identity tests means proving circuit lower
  bounds.
\newblock {\em Computational Complexity}, 13(1-2):1--46, 2004.

\bibitem[MW18]{MurrayW18}
Cody Murray and R.~Ryan Williams.
\newblock Circuit lower bounds for nondeterministic quasi-polytime: an easy
  witness lemma for {NP} and {NQP}.
\newblock In {\em Proceedings of the 50th Annual {ACM} {SIGACT} Symposium on
  Theory of Computing, {STOC} 2018, Los Angeles, CA, USA, June 25-29, 2018},
  pages 890--901, 2018.

\bibitem[Spi96]{lin-codes}
Daniel~A. Spielman.
\newblock Linear-time encodable and decodable error-correcting codes.
\newblock {\em {IEEE} Trans. Information Theory}, 42(6):1723--1731, 1996.

\bibitem[Tam16]{Tamaki16}
Suguru Tamaki.
\newblock A satisfiability algorithm for depth two circuits with a
  sub-quadratic number of symmetric and threshold gates.
\newblock {\em Electronic Colloquium on Computational Complexity {(ECCC)}},
  23:100, 2016.

\bibitem[Wag86]{Wagner86}
Klaus~W. Wagner.
\newblock The complexity of combinatorial problems with succinct input
  representation.
\newblock {\em Acta Inf.}, 23(3):325--356, 1986.

\bibitem[Wil13]{Williams-SICOMP13}
Ryan Williams.
\newblock Improving exhaustive search implies superpolynomial lower bounds.
\newblock {\em SIAM Journal on Computing}, 42(3):1218--1244, 2013.

\bibitem[Wil14]{Williams-jacm14}
Ryan Williams.
\newblock Nonuniform {ACC} circuit lower bounds.
\newblock {\em J. {ACM}}, 61(1):2:1--2:32, 2014.

\bibitem[Wil18a]{Williams-TOC18}
R.~Ryan Williams.
\newblock New algorithms and lower bounds for circuits with linear threshold
  gates.
\newblock {\em Theory of Computing}, 14(1):1--25, 2018.

\bibitem[Wil18b]{acc-algo}
R.~Ryan Williams.
\newblock New algorithms and lower bounds for circuits with linear threshold
  gates.
\newblock {\em Theory of Computing}, 14(1):1--25, 2018.

\bibitem[Wil18c]{Williams-CCC18}
Richard~Ryan Williams.
\newblock Limits on representing boolean functions by linear combinations of
  simple functions: Thresholds, relus, and low-degree polynomials.
\newblock In {\em 33rd Computational Complexity Conference, {CCC} 2018, June
  22-24, 2018, San Diego, CA, {USA}}, pages 6:1--6:24, 2018.

\end{thebibliography}

\end{document}